\documentclass[smallextended]{svjour3}

\usepackage{graphicx}              
\usepackage{amsmath}               
\usepackage{amsfonts}              
\usepackage{booktabs}
\usepackage{multirow}
\usepackage{tikz}               
\usetikzlibrary{arrows,positioning}
\usetikzlibrary{calc}

\DeclareMathAlphabet{\mathbmit}{OML}{cmm}{b}{it}
\renewcommand{\vec}[1]{\mathbmit{#1}}
\let\matr\vec

\begin{document}

\nocite{*}

\title{On a linearization technique for solving the quadratic set covering problem and variations}
\titlerunning{On a linearization technique for quadratic set covering problem}

\author{Pooja Pandey $^{\rm a}$\thanks{$^{\rm a}$ Corresponding author. Email: poojap@sfu.ca} and Abraham P. Punnen$^{\rm b}$\thanks{$^{\rm b}$This work was supported by an NSERC discovery grant awarded to Abraham Punnen,  Email: apunnen@sfu.ca}
}

\authorrunning{P. Pandey and A. P. Punnen}

\institute{Pooja Pandey \at
              Department of Mathematics, Simon Fraser University \\
              250 - 13450 – 102nd Avenue, Surrey, BC, V3T 0A3, Canada\\
              \email{poojap@sfu.ca}           
           \and
           Abraham P. Punnen \at
              Department of Mathematics, Simon Fraser University \\
              250 - 13450 – 102nd Avenue, Surrey, BC, V3T 0A3, Canada\\
              \email{apunnen@sfu.ca}
}

\date{Received: date / Accepted: date}

\maketitle

\begin{abstract}
  In this paper we identify various inaccuracies in the paper by R. R. Saxena and S. R. Arora, {\it A Linearization technique for solving the Quadratic Set Covering Problem}, Optimization, 39 (1997) 33-42. In particular, we observe  that their algorithm does  not guarantee optimality, contrary to what is claimed. Experimental analysis has been carried out to assess the value of this algorithm  as a heuristic. The results disclose that for some classes of problems the Saxena-Arora algorithm is effective in achieving good quality solutions while for some  other classes of problems, its performance is  poor. We  also discuss similar   inaccuracies in another related paper.
\end{abstract}


\section{Introduction}

The set covering problem is  well studied  in the Operations Research literature \cite{balas,baz,bec,gar,lem}. Most of the works on the problem reported in the literature have a linear objective function.  Bazaraa and Goode~\cite{baz} introduced the quadratic set covering problem (QSP) and proposed a cutting plane algorithm to solve it.  Adams \cite{adam} and Liberti \cite{leo} proposed linearization techniques for binary quadratic programs.  Since QSP is a binary quadratic programming problem,  these linearization techniques can be used to formulate QSP as a 0-1 integer linear program. QSP is known to be NP-hard  and polynomial time approximation algorithms are also available \cite{bruno} to solve  special classes of this problem.\\

 Saxena and Arora \cite{saxe}  studied QSP  and discussed various structural properties of the problem along with a linearization algorithm which is claimed to produce an optimal solution.  The notion of linearization used in \cite{saxe}  is  different from the concept of "linearization"  used by Adams \cite{adam} and Liberti \cite{leo} and also different from what is  discussed in \cite{{2015ante,2011Kabadi,2013Punnen}}. \\

In this paper, we show that the properties of QSP established in~\cite{saxe} are incorrect and that the algorithm they proposed need not produce an optimal solution. Gupta and Saxena~\cite{gupta} extended the results of~\cite{saxe} to the quadratic set packing and partitioning problems. These extensions also suffer from the same drawbacks as that of~\cite{saxe} and  the algorithm in \cite{gupta} could also  produce a non-optimal solution, contrary to what is claimed.  Since the algorithm of ~\cite{saxe} is not  guaranteed to produce  an optimal solution, it will be interesting to  examine  its value as a heuristic.  Our experimental analysis discloses that the  algorithm of ~\cite{saxe} produces good solutions for some classes of problems while it produces very poor solutions for other classes.\\

\section{The quadratic set covering problem} \label{qsp}

Let $I = \{1,2,\dotsc,m\}$ be a finite set and  $P = \{P_1,P_2,\dots,P_n\}$  be a family of subsets of $I$. The index set for  the elements of $P$ is denoted by $J= \{1,2 ,\dotsc, n\} $. For each element $j \in J$, a cost $c_j$ is prescribed and for each element $(i,j) \in J \times J$, a cost $d_{ij}$ is also prescribed. We refer to $c_j$  the  {\it linear cost} of the set $P_j$ and $\vec{c} = (c_1,\dotsc,c_n)$  the {\it linear cost vector}. Similarly $d_{ij}$ is referred to as the {\it quadratic cost} corresponding to the ordered pair $(P_i,P_j)$  and the matrix $\matr{D} =(d_{ij})_{n\times n}$ is referred  to as the {\it quadratic cost matrix}.\\

A subset $V$ of $J$ is said to be a {\it cover} of $I$, if {$\displaystyle \cup _{j \in V} P_j = I$}.
Then the {\it linear set covering problem} (LSP) is to find a cover $ L=\{\pi(1), \dotsc, \pi(l)\} $ such that  $\sum_{i=1}^{l} c_{\pi(i)}$ is minimized. Likewise the {\it quadratic set covering problem} (QSP) is to select a cover $L= \{\sigma(1), \dotsc, \sigma(l)\}$ such that $\sum_{i=1}^{l} c_{\sigma(i)} + \sum_{i=1}^{l} \sum_{j=1}^{l} d_{\sigma(i)\sigma(j) }$ is minimized.\\

 For each $i\in I$, consider the vector $\vec{a}_i=(a_{i1}, a_{i2}, \hdots, a_{in})$ where
\begin{equation*}
a_{ij} =
\begin{cases}
1 &\text{if $i\in P_j$}\\
0 &\text{otherwise.}
\end{cases}
\end{equation*}
and $\matr{A}= (a_{ij})_{m \times n}$ be an $m \times n$ matrix. Also, consider the decision variables $x_1,x_2,\ldots ,x_n$ where

\begin{equation*}
x_j =
\begin{cases}
1 &\text{if set $P_j$ is selected}\\
0 &\text{otherwise.}
\end{cases}
\end{equation*}

The vector of decision variables is represented by $ \vec{x} =  (x_1, \hdots, x_n )^T$ and $\vec{1}$ is a column vector of size $m$ where all entries are equal to 1. Then the LSP and QSP can be formulated respectively as  0-1 integer programs

\begin{align}
\nonumber \mbox{LSP:\hspace{1cm}Minimize~ } & \vec{c}\vec{x}\\
\label{eeq1} \mbox{Subject to } & \matr{A}\vec{x} \geq \vec{1}\\
\label{eeq2} &\vec{x} \in \{0,1\}^n
\end{align}
and
\begin{align}
\nonumber \mbox{QSP:\hspace{1cm}Minimize~ } & \vec{c}\vec{x} + \vec{x}^T\matr{D}\vec{x}\\
\label{eeq3} \mbox{Subject to } & \matr{A}\vec{x} \geq \vec{1}\\
\label{eeq4} &\vec{x} \in \{0,1\}^n
\end{align}

As indicated in \cite{saxe} the continuous relaxation of  QSP, denoted  by  $QSP^{'}$,  is obtained by replacing the constraint $\vec{x} \in \{0,1\}^n$ by $\vec{x} \ge \vec{0}$, where $\vec{0}$ is the zero vector of size $n$.  The family of feasible solutions of both LSP and QSP is denoted by $\bar{S} = \{ \vec{x}  | A\vec{x}\ge \vec{1}, \vec{x} \in \{0,1\}^n \}$.\\

The following definitions are taken directly from \cite{saxe}.  Any $\vec{x}\in \bar{S}$ is called a \textit{cover solution} and an optimal solution to the underlying problem (LSP or QSP)  is called an \textit{optimal cover solution}. Note that each cover solution corresponds to a cover and vice versa. A \textit{cover} $V$ is said to be redundant if $V - \{j\}$  for $j \in V$ is also a \textit{cover}.
A cover which is not \textit{redundant} is called a \textit{prime cover}.
The incidence vector $\vec{x}$ that corresponds to a \textit{prime cover} is called a \textit{prime cover solution}.\\

Garfinkel and Nemhauser \cite{gar} proved that  if the objective function in LSP has a finite optimal value then there exists a prime cover solution for which this value is attained whenever $\vec{c} \ge \vec{0}$.\\

Saxena and Arora claimed an extension of this result to $QSP^{'}$,  assuming $\vec{c} \ge \vec{0}$ and $\matr{D}$ is symmetric and positive semi-definite. More precisely, they claimed:

\begin{theorem}(Theorem 3 of \cite{saxe})\label{thm1}
 If the objective function in $QSP^{'}$ has finite optimal value then there exists a prime cover solution where this value is attained.
\end{theorem}

This result however is not true as indicated  by the following example. Let

\[{\begin{array}{ccc}
\vec{c}=(0,0,0),

&

  \matr{A}=
  \begin{pmatrix}
   1  & 1 & 0 \\
1 & 0  & 1 \\
  \end{pmatrix} \mbox{ and }

  &

  \matr{D}=
  \begin{pmatrix}
   2  & -1 & -1 \\
-1 & 1  & 0 \\
-1 & 0 & 1\\
  \end{pmatrix}

  \end{array} }
\]
Note that  $\matr{D}$ is a symmetric and positive semi-definite matrix. For the QSP and $QSP^{'}$ with $\matr{A},\matr{D}$ and $\vec{c}$ defined as above, it can be verified that  $ \vec{x}^*= (1,1,1)^T$ is an optimal solution  with the objective function value zero for both problems. The optimal \textit{cover} corresponding to $\vec{x}^*$ is $V^*= \{1,2,3\}$ which is a \textit{redundant cover} since $ V^* - \{2\} = \{1,3 \}$ is also a cover. All other cover solutions and their respective objective function values are  listed below:
\begin{equation*}
\begin{aligned}
& \vec{x}^1 = (1,0,1)^T \mbox{ redundant cover solution} \quad &f(\vec{x}^1) = 1\\
&  \vec{x}^2 = (1,1,0)^T  \mbox{ redundant cover solution } \quad &f(\vec{x}^2) = 1 \\
 &   \vec{x}^3 = (0,1,1)^T \mbox{ prime cover solution} \quad &f(\vec{x}^3) = 2 \\
&  \vec{x}^4 = (1,0,0)^T \mbox{ prime cover solution}  \quad & f(\vec{x}^4) = 2\\
\end{aligned}
\end{equation*}

None of these corresponds to an optimal solution for QSP or $QSP^{'}$. In particular, no prime cover solution is optimal for the instances of QSP and $QSP^{'}$ constructed above,  contradicting Theorem \ref{thm1}. This example also shows that Theorem \ref{thm1} cannot be corrected by replacing $QSP^{'}$ with QSP in the theorem. \\


 We now show that a variation of Theorem \ref{thm1} is  true, which relaxes the requirement of $\matr{D}$ being positive semi-definite while sign restrictions are imposed on its elements. This is summarized in our next theorem.

\begin{theorem} \label{thm1.1}
There always exists a prime cover optimal solution for QSP if  $\vec{c}$ and $\matr{D}$ are non-negative.
\end{theorem}
\begin{proof}  Let $\vec{x}^0 \in \bar{S}$ be an optimal solution of QSP. Then the corresponding optimal objective function value is

\begin{align*}
f(\vec{x}^0 ) = \vec{c} \vec{x}^0 + \vec{x}^{0^T} \matr{D} \vec{x}^0
\end{align*}

Let $J_o$ be the  cover  corresponding to the solution $\vec{x}^0$. If $J_o$ is a prime cover then statement of the theorem is correct. Otherwise we can construct a prime cover, let say $J_1$, from $J_o$ by dropping  the redundant columns. Let  $\vec{x}^1$ be the solution of  QSP with respect to the prime cover $J_1$ and

\begin{align*}
f(\vec{x}^1 ) = \vec{c} \vec{x}^1 + \vec{x}^{1^T} \matr{D} \vec{x}^1.
\end{align*}

Since $\vec{c}$ and $\matr{D}$ are non-negative and $J_1 \subset J_0$,

\begin{align*}
  f(\vec{x}^0 ) & \geq f(\vec{x}^1 )
\end{align*}

Since  $\vec{x}^0$ is an optimal solution to QSP, $f(\vec{x}^0 ) = f(\vec{x}^1 )$ and the and the result follows.

\end{proof}

  The family of feasible solutions for  continuous relaxations of LSP and QSP  is represented by $S = \{ \vec{x}  | A\vec{x}\ge \vec{1}, \vec{x} \ge \vec{0} \}$. The continuous relaxation of LSP is denoted by $LSP^{'}$. \\

Saxena and Arora~\cite{saxe} also proposed an algorithm to solve QSP and claimed that it will produce an optimal solution. Their algorithm is re-stated here.\\

\begin{itemize}\setlength{\itemindent}{1em}
\item[{\bf The}] {\bf Saxena-Arora algorithm for QSP}
\item[{\bf Step 1:}] From the  QSP, construct the corresponding $QSP^{'}$
\item[{\bf Step 2:}] Choose a feasible solution $\vec{ x}^0 \in S $ such that $ \nabla f (\vec{x}^0) \neq \vec{0}$ and
form the corresponding linear programming problem $LSP^{'}$ as

\begin{equation}
\mbox{$LSP^{'}$} \quad  \mbox{ Minimize  }_{\vec{x} \in S} \nabla f (\vec{x}^0)^T\vec{x}. \label{alg1}
\end{equation}

On solving ($LSP^{'}$), let $\vec{x}^1$, be its optimal solution.
Let $S^1 = \{\vec{x}^1\}$.

\item[{\bf Step 3:}] Starting with the point $\vec{x}^1$, form the corresponding $LSP^{'}$, and let
its optimal solution be $\vec{x}^2 \neq \vec{x}^1$. Update $S^1$ i.e. $S^1 = \{\vec{x}^1, \vec{x}^2\}$.

\item[{\bf Step 4:}] Repeat \textbf{Step 3} for the point $\vec{x}^2$, and suppose at the $i^{th}$ stage
$S^1 = \{\vec{x}^1,\vec{x}^2,\hdots,\vec{x}^i\}$. Stop, if at the $(i + 1)^{th}$ stage $\vec{x}^{i+1} \in  S^1$ ,
then $\vec{x}^{i+1}$, is the optimal solution of $QSP^{'}$.

\item[{\bf Step 5:}] If $\vec{x}^{i+1}$ is an optimal solution of the form 0 or 1 then it is a
solution of QSP otherwise, go to \textbf{Step 6}.

\item[{\bf Step 6:}] Apply Gomory cuts to find a solution of the 0 or 1 form and
the corresponding prime cover.
\end{itemize}

The algorithm discussed above suffers from various  drawbacks as listed below.

\begin{itemize}
\item[1.] Even if $\vec{c} \geq \vec{0}$, and  $\matr{D}$ is symmetric and positive semi-definite, the $LSP^{'}$ in \textbf{Step 2} could be unbounded  and hence it need not  have an  optimal solution for all instances.

\item[2.] Suppose that we apply the algorithm only for instances where $LSP^{'}$ in \textbf{Step 2} is bounded in all iterations. Even then, the solution produced in  \textbf{Step 4}  could be non-optimal to $QSP^{'}$.

\item[3.] If the algorithm terminates in \textbf{Step 5} the resulting solution could be non-optimal to QSP.

\item[4.] If the algorithm successfully  moves to \textbf{Step 6}, then also the solution produced could be non-optimal.
\end{itemize}

We now illustrate  each of the drawbacks discussed above  using counterexamples.

\begin{itemize}
\item[1.] Since $\matr{D}$ is a positive semi-definite matrix and $\vec{c} \geq \vec{0}$, the objective function value of $QSP^{'}$ is bounded below by zero. However,  $LSP^{'}$ in \textbf{Step 2} or in \textbf{Step 3} need not be bounded below. Let

\[ {\begin{array}{ccc}
\vec{c}=(0,0 ,0,0),

&

\matr{A}=
  \begin{pmatrix}
 1  & 1 & 0 & 0 \\
1 & 0  & 1 & 0 \\
1 & 0  & 0 & 1 \\
  \end{pmatrix} \mbox{ and }

  &

  \matr{D}=
  \begin{pmatrix}
   10 & -3 & -4 & -4 \\
 -3 &2  & 1 & 1\\
-4 &1 &3& 1\\
-4 &1 & 1& 3\\
  \end{pmatrix}

  \end{array} }
\]
Note that $\matr{D}$ is symmetric and positive semi-definite. Consider the instance of QSP with the above values for $\vec{c}$, $\matr{D}$, and $\matr{A}$. Starting with the feasible solution $\vec{x}^0=( 1,0,0,0)^T \in S $ of $QSP^{'}$, we get the $LSP^{'}$ in \textbf{Step 2} as

\begin{align*}
\mbox{Minimize~~~ } &\nabla f (\vec{x}^0)^T\vec{x} = 20 x_1 -6 x_2 -8 x_3   -8 x_4 \\
\mbox{Subject to: } & x_1+x_2 \ge 1\\
    & x_1+x_3 \ge 1\\
    & x_1+x_4 \ge 1\\
    & x_j \geq 0 \mbox{ for } j = 1, 2,3,4.
\end{align*}

This problem is unbounded. Thus the algorithm can not be applied in this case. The immediate conclusion  is that the Saxena-Arora ~\cite{saxe} algorithm is potentially   applicable only to  those QSP instances where the resulting   $LSP^{'}$ is bounded in every step.\\

\item[2.] The algorithm can  fail in \textbf{Step 4}.  The Saxena-Arora algorithm claims to produce an optimal solution of $QSP^{'}$ in   \textbf{Step 4}  but  this may not be true always. Consider the data

\[ {\begin{array}{ccc}
\vec{c}=(0,0 ,0,0),

&

\matr{A}=
  \begin{pmatrix}
 1  & 1 & 0 & 0 \\
1 & 0  & 1 & 0 \\
1 & 0  & 0 & 1 \\
  \end{pmatrix} \mbox{ and }

  &

  \matr{D}=
  \begin{pmatrix}
  10 & 2& 2 & 2 \\
 2 &3  & 1 &1\\
2 &1 &3& 1\\
2 &1 & 1& 4\\
  \end{pmatrix}

  \end{array} }
\]

Note that  $\matr{D} $ is symmetric and positive semi-definite. Consider the instance
of QSP with the above values for $\vec{c}$, $\matr{D}$, and $\matr{A}$, we get the QSP as

\begin{align*}
\mbox{Min } & f( \vec{x})\! = \! 10x_1^2+ \! 3x_2^2+ 3x_3^2 \!+\! 4x_4^2 \! + \! 4 x_1x_2 \! + \! 4 x_1x_3 \! + \! 4x_1x_4+ 2x_2x_3 \!+\! 2x_2x_4 \!+ \!2x_3x_4\\
\mbox{st: } & x_1+x_2 \ge 1\\
    & x_1+x_3 \ge 1\\
    & x_1+x_4 \ge 1\\
    & x_j \in \{0,1\} \mbox{ for } j = 1, 2,3,4.
\end{align*}

and

\begin{align*}
\nabla f (\vec{x}) = &(20x_1 + 4x_2+ 4x_3+ 4x_4, 6x_2 + 4x_1+2x_3+2x_4, \\
& 6x_3+4x_1+ 2x_2 + 2x_4, 8x_4+4x_1+2x_2+2x_3)^T
\end{align*}

\noindent  Select the feasible solution $\vec{x}^0= ( 1, 0, 0, 0 )^T \in S $  of $QSP^{'}$. Construct the $LSP^{'}$ with respect to $\vec{x}^0$,   the objective function of $LSP^{'}$  is

\begin{align*}
\nabla f (\vec{x}^0)^T\vec{x} = &20 x_1 +4 x_2 +4 x_3  + 4 x_4
\end{align*}

Note that $\vec{x}^1 = (0, 1, 1, 1)^T $ is an optimal solution to this $LSP^{'}$. Thus, we set  $S^1=\{\vec{x}^1\}$.  Now, using $\vec{x}^1$ construct the new $LSP^{'}$, and the optimal solution to this $LSP^{'}$ is $\vec{x}^2= ( 1, 0, 0, 0)^T$. Since $\vec{x}^2 \not\in S^1$, construct the new $LSP^{'}$ , the optimal solution to this $LSP^{'}$ is $\vec{x}^3= ( 0, 1, 1, 1)^T$.  Since $\vec{x}^3 \in S^1$,  in \textbf{Step 4} the algorithm concludes that $\vec{x}^3$ is an optimal solution of $QSP^{'}$ with objective function value  16. However, $\vec{x}^* = ( 0.714286, 0.285714, 0.285714, 0.285714 )^T $ which is a better solution for  $QSP^{'}$, contradicting the optimality of $x^3$. Thus the algorithm could fail in \textbf{Step 4}.\\

\item[3.] As per the Saxena-Arora algorithm, \textbf{Step 5} produces an optimal solution to $QSP^{'}$ and if this optimal solution is  binary,  they claim this solution to be an optimal solution of QSP.  We now show that a binary solution produced in \textbf{Step 5} need not  be an optimal solution to QSP. For example. \\

Consider the data

\[ {\begin{array}{ccc}
\vec{c}=(0,0 ,0,0),

&

\matr{A}=
  \begin{pmatrix}
 1  & 1 & 0 & 0 \\
1 & 0  & 1 & 0 \\
1 & 0  & 0 & 1 \\
  \end{pmatrix} \mbox{ and }

  &

  \matr{D}=
  \begin{pmatrix}
   4 & 1& 1 & 1 \\
 1 &2  & 0 &0\\
1 &0 &2& 0\\
1 &0 & 0& 2\\
  \end{pmatrix}

  \end{array} }
\]

Note that $\matr{D}$ is a symmetric, positive semi-definite and  non-negative. As noted  in Theorem \ref{thm1.1}, a prime cover optimal solution exists for this QSP. But still the Saxena-Arora algorithm fails to produce  an optimal solution for QSP. Consider the instance of QSP with the above values for $\vec{c}$, $\matr{D}$, and $\matr{A}$. \\

\noindent  Select the feasible solution $\vec{x}^0=( \frac{3}{4},\frac{1}{4},\frac{1}{4},\frac{1}{4})^T \in S $ which is also an optimal solution of $QSP^{'}$. Construct the $LSP^{'}$ with respect to $\vec{x}^0$ and the objective function is $\nabla f (\vec{x}^0)^T\vec{x} = \frac{15}{2} x_1 + \frac{5}{2} x_2 + \frac{5}{2} x_3  + \frac{5}{2} x_4$.\\

$\vec{x}^1= (0,1,1,1)^T$ is an optimal solution to this $LSP^{'}$. Thus, we set  $S^1=\{\vec{x}^1\}$.  Now, using $\vec{x}^1$, construct the new $LSP^{'}$ with the objective function as $\nabla f (\vec{x}^1)^T\vec{x} = 6 x_1 +4 x_2 +4x_3   +4 x_4$.  An optimal solution to this $LSP^{'}$ is $\vec{x}^2= (1,0,0,0)^T$. Since $\vec{x}^2 \not\in S^1$, we update $S^1=\{\vec{x}^1, \vec{x}^2\}$.
 Starting with $\vec{x}^2$, construct the $LSP^{'}$ with the objective function $\nabla f (\vec{x}^2)^T\vec{x}= 8 x_1 +2 x_2 +2 x_3   +2 x_4$. An optimal solution to this $LSP^{'}$ is  $\vec{x}^3= (0,1,1,1)^T$.
 Since $\vec{x}^3 \in S^1$, the algorithm concludes that $\vec{x}^3$ is an optimal solution of $QSP^{'}$. Since $\vec{x}^3$ contains 0 and 1 entries only, as per the algorithm, it is an optimal  solution to QSP and the corresponding objective function value is $6$. \\

However $\vec{x}^* = (1,0,0,0)^T$ is a better solution to the QSP with objective function value $f(\vec{x}^*) = 4$. Thus, the solution produced by the the Saxena-Arora algorithm for the above instance of QSP is not optimal.\\

In the previous example if $\vec{x}^0=(0,1,1,1)^T$ is selected instead of $( \frac{3}{4},\frac{1}{4},\frac{1}{4},\frac{1}{4})^T$, the algorithm produces $x_1=(1,0,0,0)$, $x_2=(0,1,1,1)$, and $x_3=(1,0,0,0)$, leading to an accurate  optimal solution $x_3=(1,0,0,0)$ to QSP. Note that $x_0= (0,1,1,1)^T$  and $( 1,0,0,0)^T$ are alternate optimal solutions of $LSP^{'}$ with the objective function $\nabla f (\vec{x}^0)^T\vec{x} = \frac{15}{2} x_1 + \frac{5}{2} x_2 + \frac{5}{2} x_3  + \frac{5}{2} x_4$. It is easy to show that trouble of the Saxena-Arora algorithm is not  because of the presence of alternate optimal solutions, leading  to   a  choice in selection. This can be demonstrated with the same example but by selecting a different starting point as given below. \\

\noindent  Select the feasible solution $\vec{x}^0=( 1,\frac{1}{2},0,0)^T \in S $  of $QSP^{'}$. Construct the $LSP^{'}$ with respect to $\vec{x}^0$ and the objective function is $\nabla f (\vec{x}^0)^T\vec{x} = 9 x_1 + 4 x_2 + 2 x_3  + 2 x_4$. $\vec{x}^1= (0,1,1,1)^T$ is the  unique optimal solution to this $LSP^{'}$ (easily verifiable by enumerating the basic feasible solutions). Thus, we set  $S^1=\{\vec{x}^1\}$.  Now, using $\vec{x}^1$, construct the new $LSP^{'}$ with the objective function as $\nabla f (\vec{x}^1)^T\vec{x} = 6 x_1 +4 x_2 +4x_3   +4 x_4$.  The unique optimal solution to this $LSP^{'}$ is $\vec{x}^2= (1,0,0,0)^T$. Since $\vec{x}^2 \not\in S^1$, we update $S^1=\{\vec{x}^1, \vec{x}^2\}$.
 Starting with $\vec{x}^2$, construct the $LSP^{'}$ with the objective function $\nabla f (\vec{x}^2)^T\vec{x}= 8 x_1 +2 x_2 +2 x_3   +2 x_4$. The unique optimal solution to this $LSP^{'}$ is  $\vec{x}^3= (0,1,1,1)^T$.
 Since $\vec{x}^3 \in S^1$, the algorithm concludes that $\vec{x}^3$ is an optimal solution of $QSP^{'}$. Since $\vec{x}^3$ contains 0 and 1 entries only, as per the algorithm, it is an optimal  solution to QSP and the corresponding objective function value is $6$. The solution produced by the the Saxena-Arora algorithm for the above instance of QSP is not optimal.\\

\item[4.] As per the Saxena-Arora algorithm, \textbf{Step 5} produces an optimal solution to $QSP^{'}$ and if this optimal solution is not binary, the algorithm proceeds to  \textbf{Step 6} where  Gomory cuts are applied to find a solution  which they claim to be an optimal solution to QSP.  We now show that \textbf{Step 6} need not  produce an optimal solution to QSP even if the solution produced in  \textbf{Step 4} is optimal for $QSP^{'}$.\\

    In point 3 we gave a counterexample where  the solution is a basic feasible solution (BFS) to $LSP^{'}$ which is binary but not optimal to $QSP^{'}$. Note that $QSP^{'}$ is a continuous quadratic problem and an optimal solution need not correspond to an extreme point. We now illustrate that if the $LSP^{'}$  solver works with any solution (such as interior point methods) and not necessarily with BFS (as in simplex method) it may be possible to get an optimal solution to $QSP^{'}$ in Step 4. For example

 Consider the instance of QSP from the previous case. $\vec{x}^0=( \frac{3}{4},\frac{1}{4},\frac{1}{4},\frac{1}{4})^T \in S $ is an optimal solution of $QSP^{'}$. Construct the $LSP^{'}$ with respect to $\vec{x}^0$ and the resulting  objective function is $\nabla f (\vec{x}^0)^T\vec{x} = \frac{15}{2} x_1 + \frac{5}{2} x_2 + \frac{5}{2} x_3  + \frac{5}{2} x_4$. The algorithm produces $x_1=( \frac{3}{4},\frac{1}{4},\frac{1}{4},\frac{1}{4})^T $, $x_2=( \frac{3}{4},\frac{1}{4},\frac{1}{4},\frac{1}{4})^T $,  leading to an accurate  optimal solution $x_2=(\frac{3}{4},\frac{1}{4},\frac{1}{4},\frac{1}{4})$ to $QSP^{'}$ and the algorithm successfully completes \textbf{Step 4}.  \\

 To apply Gomory cut,first reduce the non-basic feasible solution (non-BFS)  to a basic feasible solution (BFS). From the previous exampl,  the optimal  non-BFS $(\frac{3}{4},\frac{1}{4},\frac{1}{4},\frac{1}{4})$ of $LSP^{'}$ to an optimal  BFS  $\vec{x}^1= (0,1,1,1)^T$  of $LSP^{'}$. Since this is binary , no cutting plane will be added and a Gomory cut phase terminates with the non-optimal solution $\vec{x}^1= (0,1,1,1)^T$ of QSP.\\

 Alternatively if we do not reduce the non-BFS to a BFS  to apply Gomory cuts, but use any Integer programming (IP) solver to compute an optimal integer solution to $LSP^{'}$  we could still  get non-optimal solution. For example: Solving the $LSP^{'}$ at $(\frac{3}{4},\frac{1}{4},\frac{1}{4},\frac{1}{4})$ for 0-1 optimal solution  we could get  $\vec{x}^1= (0,1,1,1)^T$ as the optimal 0-1 solution of $LSP^{'}$. This is not an optimal solution to QSP. (We note that the paper \cite{saxe} does not say anything about the use of general IP solver; but we mentioned it here for the clarity and completeness).\\

\end{itemize}

\section{The quadratic set packing and partitioning problems }

 A subset $H$ of $J$ is said to be a {\it pack } of $I$ if $\bigcup_{j \in H} P_j = I$, and for $j, k \in H$, $j \neq k$, implies $P_j \bigcap P_k = \emptyset.$  Then the {\it linear set packing problem} (LSPP) is to select a {\it pack } $ V=\{\pi(1), \dotsc, \pi(v)\} $ such that  $\sum_{i=1}^{v} c_{\pi(i)}$ is maximized. Likewise, the {\it quadratic set packing problem} (QSPP) is to select a { \it pack } $L= \{\sigma(1), \dotsc,\sigma(l)\}$ such that $\sum_{i=1}^{l} c_{\sigma(i)} + \sum_{i=1}^{l} \sum_{j=1}^{l} d_{\sigma(i)\sigma(j) }$ is maximized.\\

Let $\matr{A}= (a_{ij})_{m \times n}$ be  as defined  in Section \ref{qsp}.  Also, consider the decision variables $x_1,x_2,\ldots ,x_n$ where

\begin{equation*}
x_j =
\begin{cases}
1 &\text{if $j$ is in the pack}\\
0 &\text{otherwise.}
\end{cases}
\end{equation*}

The vector of decision variables is represented as $ \vec{x} =  (x_1, \hdots, x_n )^T$ . Then the LSPP and QSPP can be formulated respectively as  0-1 integer programs

\begin{align}
\nonumber \mbox{LSPP:\hspace{1cm}Maximize~ } & \vec{c}\vec{x}\\
\label{eeq5} \mbox{Subject to } & \matr{A}\vec{x} \leq \vec{1}\\
\label{eeq6} &\vec{x} \in \{0,1\}^n
\end{align}
and
\begin{align}
\nonumber \mbox{QSPP:\hspace{1cm}Maximize~ } & \vec{c}\vec{x} + \vec{x}^T\matr{D}\textbf{x}\\
\label{eeq7} \mbox{Subject to } & \matr{A}\vec{x} \leq \vec{1}\\
\label{eeq8} &\vec{x} \in \{0,1\}^n
\end{align}

The continuous relaxations of LSPP and QSPP, denoted respectively by LSPP(C) and QSPP(C),  are obtained by replacing the constraint $\vec{x} \in \{0,1\}^n$ by $\vec{x} \ge \vec{0}$, respectively in LSPP and QSPP.\\

The family of feasible solutions of both LSPP and QSPP is denoted by $S = \{ \vec{x}  | A\vec{x}\le \vec{1}, \vec{x} \in \{0,1\}^n \}$ and the family of feasible solutions for their continuous relaxations is denoted by $\bar{S} = \{ \vec{x}  | A\vec{x}\le \vec{1}, \vec{x} \ge \vec{0} \}$.  \\

Following are some definitions given in \cite{gupta}.  A solution $\vec{x} \in S$ which satisfies (\ref{eeq7}) and  (\ref{eeq8})  is said to be a \textit{pack solution}. For any \textit{pack} $V$, a column of $\matr{A}$ corresponding to $j^* \in V $ is said to be redundant if $V - \{j^*\}$ is also a \textit{ pack}. If a \textit{pack} corresponds to  one or more redundant columns, it is called a \textit{redundant pack}. A \textit{pack} $ V^*$ is said to be a \textit{prime pack}, if none of the columns corresponding to $j^* \in  V^*$ is \textit{redundant}. A solution corresponding to the \textit{prime pack} is called a \textit{prime packing solution}.\\

From the  definition of a redundant column given above (as in \cite{gupta}),  zero vector is the only prime packing solution for the set packing problem. Thus the results of \cite{gupta} are incorrect with respect to their definitions. We believe the ``-" sign in the above definition of redundant column discussed in \cite{gupta} is a typo and it is probably supposed to be $``\cup"  $ which is consistent with the  definitions given in \cite{gupta1} by the same authors. Hereafter, we use this modified definition. \\

Thus, for any \textit{pack} $V$, a column of $\matr{A}$ corresponding to  $j \in J $ is said to be redundant if $V \cup \{j\}$ is also a pack. If a \textit{pack} contains one or more redundant columns, it is called a \textit{redundant pack}. A \textit{pack} $ V^*$ is said to be a \textit{prime pack}, if none of the columns corresponding to $j \in  J$ is \textit{redundant}. A solution corresponding to the \textit{prime pack} is called a \textit{prime packing solution}.\\

Gupta and Saxena \cite{gupta} assumed $\matr{D}$ to be a negative semi-definite matrix and extended most of the results  for QSP in  \cite{saxe} to QSPP. In particular, they claimed that: \\

\begin{theorem}[Theorem 2 of \cite{gupta}] \label{thm2}
 If the objective function in QSPP has finite value then there exists a prime packing solution where this value is attained.
\end{theorem}

Because of the definition of the  prime pack solution given by Gupta and Saxena \cite{gupta}, a prime pack is always a zero vector hence the theorem is given incorrect.  The theorem is still incorrect even if we use the  modified definition \cite{gupta1} which is indicated above.\\

For example, consider an instance of QSPP with

\[{\begin{array}{ccc}
 \vec{c}=(0,0,0),

&

  A=
  \left[ {\begin{array}{ccc}
   1  & 1 & 0 \\
1 & 0  & 1 \\
  \end{array} } \right],

  &

  D=
  \left[ {\begin{array}{ccc}
   -2  & 1 & 0 \\
1 & -2  & 1 \\
0 & 1 & -2\\
  \end{array} } \right]

  \end{array} }
\]
Note that $\matr{D}$ is symmetric and negative semi-definite.\\

  $\vec{x}^* = ( 0,0,0)^T$ is  an optimal solution for the QSPP with the objective function value zero. We list below all prime pack  solutions with the objective function values.

  \begin{equation*}
\begin{aligned}
& \vec{x}^1 = \{1,0,0\} \mbox{ prime pack solution} \quad &f(\vec{x}^1) = -2\\
 &   \vec{x}^2 = \{0,1,1\} \mbox{ prime pack solution} \quad &f(\vec{x}^2) = -2 \\
\end{aligned}
\end{equation*}

Note that none of these solutions are optimal.\\

We now show that a variation of Theorem \ref{thm2} is  true and  this is summarized in our next theorem.

\begin{theorem} \label{thm2.1}
There always exists a prime pack optimal solution for QSPP if   all elements of $\vec{c} $ and   $\matr{D}$ are non-negative.

\end{theorem}
\begin{proof}  Let $\vec{x}^0 \in \bar{S}$ be an optimal solution of QSPP. Then the corresponding optimal objective function value is

\begin{align*}
f(\vec{x}^0 ) = \vec{c} \vec{x}^0 + \vec{x}^{0^T} \matr{D} \vec{x}^0
\end{align*}

Let $J_o$ be the  pack  corresponding to the solution $\vec{x}^0$. If $J_o$ is a prime pack then we are done. Otherwise we can construct a prime pack, let say $J_1$, from $J_o$ by adding  the redundant columns. Let  $\vec{x}^1$ be the solution of  QSP with respect to the prime pack $J_1$ and

\begin{align*}
f(\vec{x}^1 ) = \vec{c} \vec{x}^1 + \vec{x}^{1^T} \matr{D} \vec{x}^1.
\end{align*}

Since $J_1$ obtained by adding redundant columns to $J_0$, therefore, $J_0 \subset J_1$. When all elements of $\vec{c} $ and  $\matr{D}$ are non-negative, or

\begin{align*}
  f(\vec{x}^0 ) & \leq f(\vec{x}^1 )
\end{align*}

Since  $\vec{x}^0$ is an optimal solution to QSPP, $f(\vec{x}^0 ) = f(\vec{x}^1 )$ and the proof follows.

\end{proof}

 Along the same lines as in \cite{saxe}, the authors of \cite{gupta}, provide a solution algorithm for QSPP. Following the insight generated in our counter examples in Section \ref{qsp}, and by  the above observation, it is not difficult to construct counter examples to show that the algorithm of \cite{gupta} need not provide an optimal solution for QSPP. \\

If in equation (\ref{eeq7}) we replace constraints $\matr{A}\vec{x} \leq \vec{1}$ with  $\matr{A}\vec{x} = \vec{1}$, then QSPP changes into quadratic set partitioning problem. Gupta and Saxena   \cite{gupta} proposed a  similar   algorithm for the quadratic set partitioning problem, which has similar issues as in the quadratic set packing problem. We omit the discussion about the quadratic set partitioning problem.

\section{Computational results}

Since the algorithm of \cite{saxe} is not  guaranteed to be optimal, it would be interesting to  examine its value  as a heuristic  to solve QSP.  We have conducted some preliminary experimental analysis  to assess the value of  the Saxena-Arora algorithm as a heuristic using different classes of test problems.\\

The test data was taken from standard benchmark problems for the set covering problem \cite{balas,beas,wool}, and the vertex covering problem \cite{kexu}, with appropriate amendments to incorporate quadratic objective. In this class, we took only small size instances since the quadratic problem is much more difficult and time consuming to solve compared to their linear counterparts. We have also generated some  quadratic vertex cover instances on random graphs taken from \cite{muth}. We divided computational experiments into two different categories, with each $\vec{c} \geq 0$, while in category 1: $\matr{D}$ is a positive semi-definite matrix and in category 2:  $ \matr{D} $ is non-negative  and positive semi-definite.\\

Each element of the linear cost vector $\vec{c}$ is a random integer from the interval [3,5]. Since the quadratic cost matrix $\matr{D} $ is positive semi-definite,   there exists a square matrix $\matr{B}$ such  that $\matr{D} = \matr{B}\matr{B}^T$. This $\matr{D} $ is generated by a random square matrix $\matr{B}$ where each element of $\matr{B}$ is a random integer between  -10 and 10. When  $ \matr{D} $ is non-negative  and positive semi-definite, each element of $\matr{B}$ is selected as a random integer between 0 and 20.\\

The Saxena-Arora algorithm was coded in C++ and tested on a PC with windows 7 operating system, Intel 3770 i7 3.40 GHz processor and with 16 GB of RAM.  We also used CPLEX  0-1 integer quadratic solver (version 12.5) to compute exact (heuristic) solutions. For each instance that we tested, we set CPLEX time limit to be the same as the time taken by Saxena-Arora algorithm and also run CPLEX by doubling this running time. These two implementations provide heuristic solutions and were compared with the solution produced by the Saxena-Arora algorithm.\\

In the tables, $t_1$ is the cpu time taken by Saxena-Arora algorithm. The column ``CPLEX Sol{ ($t_1$})" represents  the heuristic solution obtained by CPLEX by fixing its running time to $t_1$ and the column ``CPLEX Sol{ ($2t_1$})" represents CPLEX run with $2t_1$ upper bound on the execution time. The column ``negative entries in Q"  provides percentage of negative entries in the matrix $\matr{D} $. The column ``Sol" refers  the objective function values. CPLEX quadratic solver takes more time to solve QSP to  optimality when $\matr{D}$ is positive semi-definite compare to the instances when $\matr{D}$ is non-negative and positive semi-definite.  Therefore,  Table \ref{table:1} reports  lower bound value and  Table \ref{table:2}  reports optimal solution value.

\begin{table}[h!]
\caption{Benchmark instances, $\matr{D}$ is positive semi-definite}
\resizebox{\columnwidth}{!}{%
\begin{tabular}{@{}l cccccccc @{}} \toprule[1.5pt]
\multirow{3}[3]{*}{problem} & \multicolumn{2}{@{}c@{}}{size} & \multirow{3}{*}{\parbox{1.5cm}{Lower bound on opt}} &\multicolumn{2}{c}{Saxena Algo.}  & \multirow{3}{*}{\parbox{1.5cm}{CPLEX Sol{ ($t_1$})}} & \multirow{3}{*}{\parbox{1.5cm}{CPLEX Sol{ ($2t_1$})}} &\multirow{3}{*}{\parbox{1.5cm}{Negative entries in D(\%)}} \\ \cmidrule[1.2pt](lr){2-3} \cmidrule[1.2pt](lr){5-6}
 	&	\multirow{2}{*}{m} &\multirow{2}{*}{n} &  & \multirow{2}{*}{\parbox{1.5cm}{cpu time $t_1$ (sec)}} & \multirow{2}{*}{Sol} &  &  &    	\\
	&	 & & & & 	 &  & 	&  \\ \midrule
Qscpcyc06& 192 & 240 &48& 156 & 241189  & 71 & 70 & 43.59 \\
Qscpcyc07 & 448 & 672 &112& 125 & 1486218  & 192 & 187 &46.00\\
Qscp41 & 1000 & 200 &432.27& 78 & 7743176  & 455 & 455 &44.48\\
Qscpe3 & 500 & 50 &3.307& 63 & 1777128 &  10 & 10 &46.46\\
Qscpe4 & 500 & 50 &3.455& 62 & 1638640  & 10 & 10 &45.69\\
Qscpe5 & 500 & 50 &3.393& 47 & 1918480  & 28 & 15 &45.69\\
Qgraph50-01 & 612 & 50 &111.75& 156 & 144  & 122 & 122 & 41.04 \\
Qgraph50-02 & 490 & 50 &93.6684& 172 & 641  & 124 & 124 & 42\\
Qgraph50-03 & 735 & 50 &101.25& 156 & 153 & 152 & 152 & 40.23\\
Qgraph50-04 & 612 & 50 &87.9275& 156 & 1057 &  149 & 149 &45.6\\
Qgraph50-05 & 490 & 50 &108.25& 156 & 127  & 124 & 124 &45.6\\
Qgraph50-06 & 857 & 50 &89.71486& 187 & 261 & 154 & 154 & 41.6 \\
Qgraph50-07 & 735 & 50 &85.1275& 156 & 1738  & 119 & 119 & 48\\
Qgraph50-08 & 612 & 50 &89.1389& 156 & 1030  & 100 & 100 & 44\\
Qgraph50-09 & 980 & 50 &89.6389& 156 & 600 &  146 & 146 & 47.52\\
Qgraph50-10 & 612 & 50 &92.5273& 156 & 2363  & 142 & 142 & 48.72\\
Qfrb30-15-1 & 17827 & 450 &810.306& 11185 & 5106  & 1494 & 1494 & 42.56 \\
Qfrb30-15-2 & 17874 & 450 &799.139& 7909 & 18906  & 1479 & 1479 & 44.59\\
Qfrb30-15-3 & 17809 & 450 &799.139& 6100 & 11497  & 1478 & 1477 & 44.59\\
Qfrb30-15-4 & 17831 & 450 &787.806& 5274 & 21930 &  1475 & 1475 &44.46\\
Qfrb30-15-5 &  17794 & 450 & 799.139 & 5616 & 14032  & 1478 & 1475 & 44.59  \\
 \bottomrule[1.5pt]
\end{tabular}
}
\label{table:1}
\end{table}

\begin{table}[h!]
\caption{Benchmark instances, $ \matr{D} $ is non-negative  and positive semi-definite}
\resizebox{\columnwidth}{!}{%
\begin{tabular}{@{}l ccccccc @{}} \toprule[1.5pt]
\multirow{3}[3]{*}{problem} & \multicolumn{2}{@{}c@{}}{size} & \multirow{3}{*}{\parbox{1.5cm}{Optimal sol}} &\multicolumn{2}{c}{Saxena Algo.}  & \multirow{3}{*}{\parbox{1.5cm}{CPLEX Sol{ ($t_1$})}} & \multirow{3}{*}{\parbox{1.5cm}{CPLEX Sol{ ($2t_1$})}}  \\ \cmidrule[1.2pt](lr){2-3} \cmidrule[1.2pt](lr){5-6}
 	&	\multirow{2}{*}{m} &\multirow{2}{*}{n} &  & \multirow{2}{*}{\parbox{1.5cm}{cpu time $t_1$ sec }} & \multirow{2}{*}{Sol} &  &     	\\
	&	 & & & & 	 &  & 	  \\ \midrule
Qscpcyc06& 192 & 240 & 147523 & 483   &  147523  & 147523  & 147523   \\
Qscpcyc07 & 448 & 672 & 726070   &  780 & 726070 & 726070  & 726070  \\
Qscp41 & 1000 & 200 & 7271   &  343 & 7271 & 7271  & 7271  \\
Qscpe3 & 500 & 50 & 7 &  390 & 3369  &  7 & 7  \\
Qscpe4 & 500 & 50 & 8 &  515 & 3726  &  8 & 8  \\
Qscpe5 & 500 & 50 & 7 & 359   & 1854  &  7 &  7 \\
Qgraph50-01 & 612 & 50 & 5149 & 187 &  5149 & 5149 & 5149  \\
Qgraph50-02 & 490 & 50 & 5108 & 156 & 5108  & 5108 & 5108 \\
Qgraph50-03 & 735 & 50 & 5163 & 188 & 5163 & 5163 & 5163 \\
Qgraph50-04 & 612 & 50 & 9272 & 156 & 9272 & 9272 & 9272 \\
Qgraph50-05 & 490 & 50 & 4168 & 125 & 4168  & 4168 & 4168 \\
Qgraph50-06 & 857 & 50 & 8042 & 124 & 8042 & 8042 & 8042  \\
Qgraph50-07 & 735 & 50 & 7088 & 124 & 7088  & 7088 & 7088 \\
Qgraph50-08 & 612 & 50 & 4670 & 109 & 4670 & 4670 & 4670 \\
Qgraph50-09 & 980 & 50 & 8872 & 141 & 8872 & 8872 & 8872\\
Qgraph50-10 & 612 & 50 & 6614 & 125 & 6614 & 6614 & 6614 \\
\bottomrule[1.5pt]
\end{tabular}
}
\label{table:2}
\end{table}

When  $\matr{D}$ is a random positive semi-definite matrix, Table \ref{table:1} shows that the Saxena-Arora algorithm does not return a good quality solution for QSP. Note that a general purpose solver like CPLEX  obtained much better  solutions within the same time limit for the test problems used. But when $ \matr{D} $ is non-negative  and positive semi-definite, Table \ref{table:2} shows that the Saxena-Arora algorithm produced solutions as good as those produced by CPLEX for many instances. For vertex cover  instances the Saxena-Arora algorithm produced an optimal solution. For the set cover instances CPLEX produces better solutions than the Saxena-Arora algorithm.  Thus, for $ \matr{D} $ is non-negative  and positive semi-definite, the Saxena-Arora algorithm  could be used as a heuristic to solve QSP. As our counter example indicates, even for this class the Saxena-Arora algorithm need not produce an optimal solution.

\end{document}